\documentclass[a4paper,UKenglish,cleveref, autoref, thm-restate]{lipics-v2021}

\title{Optimal Construction of Hierarchical Overlap Graphs} 
\titlerunning{Optimal Construction of HOG} 

\author{Shahbaz Khan}{University of Helsinki, Finland}{shahbaz.khan@helsinki.fi}{https://orcid.org/0000-0001-9352-0088}{}

\authorrunning{S. Khan} 
\Copyright{Shahbaz Khan} 


\category{} 

\relatedversion{Accepted to be published in CPM 2021}

\funding{This work was  funded by the European Research Council (ERC) under the European Union's Horizon 2020 research and innovation programme (grant agreement No.~851093, SAFEBIO).}

\acknowledgements{I would like to thank Alexandru I. Tomescu for helpful discussions, and for critical review and insightful suggestions which helped me in refining the paper. I would also like to thank Veli Mäkinen for pointing out the similarity with the classical result for APSP problem.}

\nolinenumbers 

\hideLIPIcs  

\EventEditors{Pawe{\l} Gawrychowski and Tatiana Starikovskaya}
\EventNoEds{2}
\EventLongTitle{32nd Annual Symposium on Combinatorial Pattern Matching (CPM 2021)}
\EventShortTitle{CPM 2021}
\EventAcronym{CPM}
\EventYear{2021}
\EventDate{July 5--7, 2021}
\EventLocation{Wroc{\l}aw, Poland}
\EventLogo{}
\SeriesVolume{191}
\ArticleNo{16}

\usepackage[vlined,ruled]{algorithm2e}
\usepackage{xcolor}

\begin{document}

\maketitle

\begin{abstract}
Genome assembly is a fundamental problem in Bioinformatics, where for a given set of overlapping substrings of a genome, the aim is to reconstruct the source genome. The classical approaches to solving this problem use assembly graphs, such as {\em de Bruijn graphs} or {\em overlap graphs}, which maintain partial information about such overlaps. For genome assembly algorithms, these graphs present a trade-off between overlap information stored and scalability. Thus, Hierarchical Overlap Graph (HOG) was proposed to overcome the limitations of both these approaches.

For a given set $P$ of $n$ strings, the first algorithm to compute HOG was given by Cazaux and Rivals~[IPL20] requiring $O(||P||+n^2)$ time using superlinear space, where $||P||$ is the cumulative sum of the lengths of strings in $P$. This was improved by Park et al.~[SPIRE20] to $O(||P||\log n)$ time and $O(||P||)$ space using segment trees, and further to $O(||P||\frac{\log n}{\log \log n})$ for the word RAM model. Both these results described an open problem to compute HOG in optimal $O(||P||)$ time and space. In this paper, we achieve the desired optimal bounds by presenting a simple algorithm that does not use any complex data structures. At its core, our solution improves the classical result [IPL92] for a special case of the All Pairs Suffix Prefix (APSP) problem from $O(||P||+n^2)$ time to optimal $O(||P||)$ time, which may be of independent interest.
\end{abstract}

\section{Introduction}
\label{sec:intro}
Genome assembly is one of the oldest and most fundamental problems in Bioinformatics~\cite{peltola83}. Due to practical limitations, sequencing an entire genome as a single complete string is not possible, rather a collection of the {\em substrings} of the genome (called {\em reads}) are sequenced. The goal of a sequencing technology is to produce a collection of reads that cover the entire genome and have sufficient overlap amongst the reads. This allows the source genome to be reconstructed by ordering the reads using this overlap information. The genome assembly problem thus aims at computing the source genome given such a  collection of overlapping reads. Most approaches of genome assembly capture this overlap information into an {\em assembly graph}, which can then be efficiently processed to assemble the genome. The prominent approaches use assembly graphs such as {\em de Bruijn graphs}~\cite{Pevzner1989} and {\em Overlap graphs} (also called string graphs~\cite{bti1114}), which have been shown to be successfully used in various practical assemblers~\cite{Velvet,Spades,metaSpades,hSpades,PevznerTW01,SimpsonD10}.

The de Bruijn graphs are built over $k$ length substrings (or $k$-mers) of the reads as nodes, and arcs denoting $k-1$ length overlaps among the $k$-mers. 
Their prominent advantage is that their size is linear in that of the input. However, their limitations include losing information about the relationship of $k$-mers with the reads, and in general not being able to represent overlaps of size other than $k-1$ among the reads (except \cite{BoucherBGPS15,BelazzouguiGMPP18,BelazzouguiC19}). On the other hand, Overlap graphs have each read as a node, and edges between every pair of nodes represent their corresponding maximum overlap. In practice, only the edges having certain threshold value of overlap are considered. Though they store more overlap information than de Bruijn graphs, they do not maintain whether two pairs of strings have the same overlap. Moreover, they are inherently quadratic in size in the worst case, and computing the edge weights  (even optimally~\cite{GusfieldLS92,TustumiGTL16,LimP17}) is difficult in practice for large data sets.

As a result, Hierarchical Overlap Graphs (HOG) were proposed in~\cite{CazauxCR16,CazauxR20} as an alternative to overcome such limitations of the two types of assembly graphs. The HOG has nodes for all the longest overlaps between every pair of strings, and edges connecting strings to their suffix and prefix, using linear space. Note that Overlap graphs have edges representing longest overlaps between strings requiring quadratic size, whereas HOG has additional nodes for longest overlaps between strings requiring linear size by exploiting pairs of strings having the same longest overlaps. Thus, it is a promising alternative to both de Bruijn graph and Overlap graph to better solve the problem of genome assembly. Also, since it maintains if two pairs of strings have the same overlap, it also has the potential to better solve the approximate {\em shortest superstring problem}~\cite{Ukkonen90} having applications in both genome assembly and data compression~\cite{Sweedyk99,BlumLTY94}. Some applications of HOG have been studied in~\cite{CazauxCR16,CanovasCR17}.

Cazaux and Rivals~\cite{CazauxR20} presented the first algorithm to build HOG efficiently. They showed how HOG can be computed for a set of $n$ strings $P$ in $O(||P||+n^2)$ time, where  $||P||$ represents the cumulative sum of lengths of strings in $P$. However, they required $O(||P||+n\times\min(n,\max_{p\in P} |p|))$ space, which is superlinear in input size. Park et al.~\cite{ParkCPR20} improved it to $O(||P||\log n)$ time requiring linear space using Segment trees~\cite{BergCKO08},  assuming a constant sized character set. For the word RAM model, they further improved it to  $O(||P||\frac{\log n}{\log \log n})$ time. For practical implementation, both these results build HOG using an intermediate Extended HOG (EHOG) which reduces the memory footprint of the algorithm.  In both the results, the \textit{bottleneck} is solving a special case of All Pairs Suffix Prefix (APSP) problem. Given a set $P$ of $n$ strings, the goal of the APSP problem is to compute the maximum overlaps between every pair of strings. This classical problem was optimally solved by  Gusfield et al.~\cite{GusfieldLS92} using $O(||P||+n^2)$ time and $O(||P||)$ space, where the solution is reported for the $n^2$ pairs. However, for computing HOG we  only require the set of {\em maximum overlaps}, and not their association with the corresponding pairs of strings, making the result suboptimal due to the extra $O(n^2)$ factor. 
Also, both these results~\cite{CazauxR20,ParkCPR20} mentioned as an open problem the construction of HOG using optimal $O(||P||)$ time and space.
We answer this open question positively and solve the special case of APSP optimally as follows.
   
\begin{restatable}[Optimal HOG]{theorem}{optHOG}
\label{thm:HOG}
    For a set of strings $P$, the Hierarchical Overlap Graph can be computed using $O(||P||)$ time and space.
\end{restatable}

Moreover, unlike~\cite{ParkCPR20} our algorithm does not use any complex data structures for its implementation. Also, we do not assume any limitations on the character set. Finally, like~\cite{CazauxR20,ParkCPR20} our algorithm can also use EHOG as an intermediate step for improving memory footprint in practice. Note that the size EHOG and HOG can even be identical for some instances, but their ratio can tend to infinity for some families of graphs~\cite{CazauxR20}. Thus, despite the existence of optimal algorithm for computing EHOG, an optimal algorithm for computing HOG is significant from both theoretical and practical viewpoints.

\subparagraph*{Note.} Another result~\cite{abs-2102-12824} simultaneously achieve the same optimal bound by reducing the problem to computing {\em borders}~\cite{KnuthMP77}. However, our result is simpler and more self-contained. 

\subparagraph*{Outline of the paper.} We first describe notations and preliminaries that are used in our paper in \Cref{sec:prelim}. In \Cref{sec:prevRes}, we briefly describe the previous approaches to compute HOG. Thereafter, \Cref{sec:algo} describes our core result in three stages for simplicity of understanding, each building over the previous, to give the optimal algorithm. Finally, we present the conclusions in \Cref{sec:conc}.

\section{Preliminaries}
\label{sec:prelim}
Given a finite set $P=\{p_1,...,p_n\}$ of $n$ non-empty strings over a finite set of characters, we denote the size of a string $p_i$ by $|p_i|$ and the cumulative size of $P$ by $||P||=\sum_{i=1}^n |p_i|$ ($\geq n$ as strings are not empty). 
For a string $p$, any substring that starts from the first character of $p$ is called a {\em prefix} of $p$, whereas any substring which ends at the last character of $p$ is called a {\em suffix} of $p$. A prefix or suffix of $p$ is called {\em proper} if it is not same as the whole $p$. For an ordered pair of string $(p_1,p_2)$, a string is called their {\em overlap} if it is both a proper suffix of $p_1$ and a proper prefix of $p_2$, where $ov(p_1,p_2)$ denotes the {\em longest} such overlap. Also, for the set of strings $P$, $Ov(P)$ denotes the set of all $ov(p_i,p_j)$ for $1\leq i,j \leq n$. An empty string is denoted by $\epsilon$.
We also use the notions of HOG, EHOG and the Aho-Corasick trie as follows. 

\begin{definition}[Hierarchical Overlap Graph~\cite{CazauxR20}]
Given a set of strings $P=\{p_1,\cdots,p_n\}$, its Hierarchical Overlap Graph is a directed graph ${\cal H}=(V,E)$, where
\begin{itemize}
    \item $V=P\cup Ov(P)\cup\{\epsilon\}$ and $E=E_1\cup E_2$, having
    \item $E_1=\{(x,y):x$ is the longest proper prefix of $y$ in $V\}$  as {\em tree edges}, and 
    \item $E_2=\{(x,y):y$ is the longest proper suffix of $x$ in $V\}$ as {\em suffix links}.
\end{itemize}
\end{definition}

The {\em extended HOG} of $P$ (referred as ${\cal E}$) is also similarly defined~\cite{CazauxR20}, having additional nodes corresponding to every overlap (not just longest) between each pair of strings in $P$, with the same definition of edges. The construction of both these structures uses the Aho-Corasick Trie~\cite{AhoC75} which is computable in $O(||P||)$ time and space. The Aho-Corasick Trie of $P$ (referred as $\cal A$) contains all prefixes of strings in $P$ as nodes, with the same definition for edges. All these structures are essentially {\em trees} having the empty string $\epsilon$ as the root, and the strings of $P$ as its {\em leaves}. A {\em tree edge} $(x,y)$ is labelled with the substring of $y$ not present in $x$.
Hence, despite being a graph due to the presence of {\em suffix links} (also called {\em failure links}),  we abuse the notions used for tree structures when applying to $\cal A,E$ or $\cal H$ (ignoring suffix links). Also, while referring to a node $v$ of $\cal A, E$ or $\cal H$, we represent its corresponding string with $v$ as well.

Consider \Cref{fig:aehog} for a comparison of  ${\cal A}$, $\cal E$ and $\cal H$ for $P=\{aabaa,aadbd,dbdaa\}$. Since $\cal A$ contains all prefixes as nodes, the tree edges have labels of a single character. However, $\cal E$ contains all overlaps  among strings of $P$, so it can potentially have fewer internal nodes ($\{a,aa,db,dbd\}$) than $\cal A$. Further, $\cal H$ contains only longest overlaps so it can potentially have even fewer internal nodes ($\{aa,dbd\}$). 

\begin{figure}
    \centering
    \includegraphics[scale=1.35]{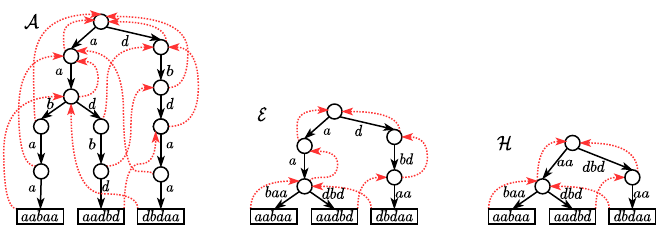}
    \caption{Given $P=\{aabaa,aadbd,dbdaa\}$, the figure shows from left to right the Aho-Corasick Trie ($\cal A$), Extended Hierarchical Overlap Graph ($\cal E$) and Hierarchical Overlap Graph ($\cal H$) of $P$.}
    \label{fig:aehog}
\end{figure}
    
Now, to compute $\cal E$ or $\cal H$ one must only remove some internal nodes from $\cal A$ and adjust the edge labels accordingly. This requires the computation of all overlaps among strings in $P$ for $\cal E$, which is further restricted to only the longest overlaps for $\cal H$. For a string $p_i\in P$ (leaf of $\cal A$), all its prefixes are its ancestors in ${\cal A}$, whereas all its suffixes are on the path following the suffix links from it (referred as {\em suffix path}).
Thus, every internal node is implicitly the prefix of its descendant leaves, and to be an overlap it must merely be a suffix of some string in $P$~\cite{Ukkonen90}. Hence to compute internal nodes of $\cal E$ (or overlap) from $\cal A$ one simply traverses the suffix paths from all the leaves of $\cal A$, and remove the non-traversed internal nodes (see \Cref{fig:aehog}). However, to compute $\cal H$ from $\cal A$ (or $\cal E$) we need to find only  the longest overlaps, which is equivalent to solving a special case of the APSP problem, requiring only the set of all maximum overlaps. We use the following criterion (also used by \cite{GusfieldLS92}) to identify the internal nodes of $\cal H$.

\begin{lemma}[\cite{GusfieldLS92}]
An internal node $v$ in ${\cal A}$ (or $\cal E$) of $P$, is  $ov(p_i,p_j)$ for two strings $p_i,p_j\in P$ iff $v$ is an overlap of $(p_i,p_j)$ and no descendant of $v$ is an overlap of $(p_i,p_j)$.
\label{lem:condition}
\end{lemma}
\begin{proof}
The ancestor of a node $v$ in ${\cal A}$ is its proper prefix and hence is shorter than $v$. Since two internal nodes of $\cal A$ which are both overlaps of $(p_i,p_j)$, are prefixes of $p_j$ and hence have an ancestor-descendant relationship, where the descendant is longer in length. Thus, the longest overlap $ov(p_i,p_j)$ cannot have a descendant which is an overlap of $(p_i,p_j)$.
\end{proof}

Hence to compute $Ov(P)$ (or nodes of $\cal H$), we need to check each internal node $v$ if it is the lowest overlap (in ${\cal A}$) for some pair $(p_i,p_j)$. This implies that $v$ is a suffix of some $p_i$, such that for some descendant leaf $p_j$, no suffix of $p_i$ is on path from $v$ to $p_j$ (see \Cref{fig:aehog}).

\section{Previous results}
\label{sec:prevRes}
Cazaux and Rivals~\cite{CazauxR20} were the first to study $\cal H$, where they used $\cal E$~\cite{CanovasCR17} as an intermediate step in the computation of $\cal H$.
They showed that $\cal E$ can be constructed in $O(||P||)$ time and space from $\cal A$~\cite{AhoC75}, which itself is computable in $O(||P||)$ time and space. In order to compute $\cal H$, the main {\em bottleneck} is the computation of $Ov(P)$ (i.e. solving APSP), after which we simply remove the internal nodes not in $Ov(P)$ from $\cal E$ (or $\cal A$), in $O(||P||)$ time and space. They gave an algorithm to compute $Ov(P)$ in $O(||P||+ n^2)$ time using $O(||P||+n\times\min(n,\max\{|p_i|\}))$ space. This procedure was recently improved by Park et al.~\cite{ParkCPR20} to require  $O(||P||\log n)$ time and $O(||P||)$ space using segment trees, assuming constant sized character set. For the word RAM model they further improve the time to $O(||P||\frac{\log n}{\log \log n})$. The main ideas of the previous results can be summarized as follows.

\paragraph*{Computing $Ov(P)$ in $O(||P||+n^2)$ time~\cite{CazauxR20}}
The algorithm computes $Ov(P)$ by considering the internal nodes in a bottom-up manner, where a node is processed after its descendants. Firstly, for each internal node $u$, they compute the list $R_l(u)$ (called ${\cal L}_u$ in our algorithm) of all leaves having $u$ as a suffix. Now, while processing a node $u$, they check whether $u=ov(v,x)$, i.e., $u$ is a suffix of some leaf $v$ such that the path to at least one of $u$'s descendant leaf (say $x$) does not have a suffix of $v$. 
To perform this task, they maintain a bit-vector for all leaves (suffix $v$), which is marked if no such descendant path exists from $u$ for such leaves. For a leaf $v$, the bit is implicitly marked if all children of $u$ have the bit for $v$ marked. Otherwise, if $v\in R_l(u)$ it is marked adding $u$ to $\cal H$, else left unmarked. The space requirement is dominated by that of this bit-vector, and it is computed only for the branching nodes, taking total $O(||P||+n^2)$ time. 

\paragraph*{Computing $Ov(P)$ in $O(||P||\log n)$ time~\cite{ParkCPR20}}
The algorithm firstly orders the strings in $P$ lexicographically in $O(||P||)$ time (requires constant sized character set). This allows them to define an interval of leaves which are the descendants of each internal node in $\cal E$. Now, for each leaf $v$ (suffix)  they start with an unmarked array corresponding to all leaves (prefix). Then starting from $v$ they follow its suffix path and at each internal node $u$, check if some descendant leaf $x$ (prefix) is unmarked. In such a case $u=ov(v,x)$ and hence $u$ is added to $\cal H$. Before moving further in the next suffix path the interval corresponding to all the descendant leaves (prefix) of $u$ is marked in the array. Since both query and update (mark) over an interval can be performed in $O(\log n)$ time using a segment tree, the total time taken is $O(||P||\log n)$ using $O(||P||)$ space.

\section{Our algorithm}
\label{sec:algo}
Our main contribution is an alternative procedure to compute $Ov(P)$ in $O(||P||)$ time and space which results in an optimal algorithm for computing $\cal H$ for $P$ in $O(||P||)$ time and space. 
Our overall approach is similar to that of the original algorithm~\cite{CazauxR20} with the exception of a procedure to mark the internal nodes that belong to $\cal H$, i.e.,  Mark${\cal H}$.
The algorithm except for the procedure Mark${\cal H}$ takes $O(||P||)$ time and space (also shown in~\cite{CazauxR20}). 
We describe our algorithm for Mark${\cal H}$ in three stages, first for a single prefix leaf requiring $O(||P||)$ time, and then for all prefix leaves requiring overall $O(||P||+n^2)$ time, and finally improving it to overall $O(||P||)$ time, which is optimal. The algorithm can be applied to any of ${\cal A}$ or $\cal E$, both computable in $O(||P||)$ time and space.

\subparagraph*{Note:} The second stage of our algorithm is equivalent to   \cite{GusfieldLS92}, and achieves the same bounds as \cite{CazauxR20} for computing $\cal H$, though using a simpler technique and linear space.

\subsection{Outline of Approach}
We first describe our overall approach in \Cref{alg:mainHOG}. After computing $\cal A$, for each internal node $v$, we compute the list ${\cal L}_v$ of all the leaves having $v$ as its suffix. As described earlier, this can be done by following the suffix path of each leaf $x$, adding $x$ to ${\cal L}_y$ for every internal node $y$ on the path. Using this information of suffix (in ${\cal L}_v$) and prefix (implicit in ${\cal A}$) we mark the nodes of $\cal A$ to be added in the HOG $\cal H$. We shall describe this procedure Mark$\cal H$ later on. Thereafter, in order to compute ${\cal H}$ we simply merge the unmarked internal nodes of $\cal A$ with its parents. This process is carried on using a DFS traversal of $\cal A$ (ignoring suffix links) where for each unmarked internal node $v$, we move all its edges to its parent, prepending their labels with the label of the parent edge of $v$.

\begin{algorithm}[tbh]
    	\caption{\textsc{Hierarchical Overlap Graphs}}
    	\label{alg:mainHOG}
    	\DontPrintSemicolon
    	\BlankLine
    	${\cal A} \gets$ Aho-Corasik Trie of $P$
    	\tcp*{Trie with suffix links}
    	\BlankLine
    	\lForEach{internal node $v$ of ${\cal A}$}{
    	${\cal L}_v\gets \emptyset$ 
    	\tcp*[f]{List of leaves with suffix $v$}}
    	
    	\ForEach(\tcp*[f]{Compute all ${\cal L}_v$})
    	{leaf $x$ of ${\cal A}$}{
    	    $y\gets$ Suffix link of $x$ in ${\cal A}$\;
    	    
    	    \While(\tcp*[f]{$\epsilon$ is the root of $\cal A$}){$y\neq \epsilon$}{
    	        Add $x$ to  ${\cal L}_y$\;
    	        $y\gets$ Suffix link of $y$ in ${\cal A}$\;
    	}}
    	\BlankLine

    	$in{\cal H}\gets $Mark${\cal H}({\cal A},{\cal L})$
        \tcp*{Procedure to mark nodes of ${\cal H}$ in flags $in{\cal H}$}

    	\BlankLine
    	\ForEach(\tcp*[f]{Compute ${\cal H}$})
    	{node $v\in {\cal A}$ in DFS order} {
        	\lIf{$in{\cal H}[v]=0$}{
        	   Merge $v$ with its parent
        	  }
    	}

	\end{algorithm}
	
As previously described, $\cal A$ can be computed in $O(||P||)$ time and space~\cite{AhoC75}.     
Computing ${\cal L}_v$ for all $v\in {\cal A}$  requires each leaf $p_i$ to follow its suffix  path in $O(|p_i|)$ time, and add $p_i$ to at most $|p_i|$ different ${\cal L}_y$, requiring total $O(||P||)$ time for all $p_i\in P$. This also limits the size of ${\cal L}_v$ for all $v\in \cal A$ to $O(||P||)$. Since merge operation on a node $v$ requires $O(deg(v))$ cost, computing $\cal H$ using $in{\cal H}$ requires total $O(|\cal A|)=O(||P||)$ time as well. Thus, we have the following theorem (also proved in \cite{CazauxR20}).

\begin{theorem}
For a set of strings $P$, the computation of Hierarchical Overlap Graph except for Mark${\cal H}$ operation requires $O(||P||)$ time and space.
\label{thm:mainHOG}
\end{theorem}

\subsection{Marking the nodes of ${\cal H}$}
We shall describe our procedure to mark the nodes of ${\cal H}$ in three stages for simplicity of understanding. First, we shall describe how to mark all internal nodes representing all longest overlaps $ov(\cdot,v)$ from a single leaf $v$ (prefix) in $\cal A$, using $O(||P||)$ time. Thereafter, we extend this to compute such overlaps from all leaves in ${\cal A}$ together using $O(||P||+n^2)$ time (equivalent to \cite{GusfieldLS92}). Finally, we shall improve this to our final procedure requiring optimal $O(||P||)$ time. All the three procedures require $O(||P||)$ space.

\begin{figure}
    \centering
    \includegraphics[scale=1.25]{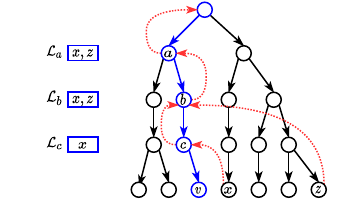}
    \caption{Overlaps of $v$ with all leaves, where $c=ov(x,v)$ and $b=ov(z,v)$ are in $Ov(P)$.}
    \label{fig:hog-stack}
\end{figure}
    
\paragraph*{Marking all nodes $ov(\cdot,v)$ for a leaf $v$}
In order to compute all longest overlaps of a leaf $v$ (see \Cref{fig:hog-stack}), we need to consider all its prefixes (ancestors in $\cal A$) according to \Cref{lem:condition}. Here the internal nodes $a,b$ and $c$ are prefixes of $v$ and also suffixes of $x$, whereas $z$ only has suffixes $a$ and $b$. Thus, we have ${\cal L}_a={\cal L}_b=\{x,z\}$ and ${\cal L}_c=\{x\}$. Thus, given that $a,b$ and $c$ are ancestors of $v$, $a$ and $b$ are valid overlaps of $(x,v)$ and $(z,v)$, whereas $c$ is only a valid overlap of $(x,v)$. Using \Cref{lem:condition}, for being the longest overlap of a pair of strings, no descendant should be an overlap of the same pair of strings. Hence, $c=ov(x,v)$ and $b=ov(z,v)$, but $a$ is not the longest overlap for any pair of strings because of $b$ and $c$. Processing ${\cal L}_u$ for all nodes on the ancestors of the leaf (prefix) requires $O(||P||)$ time. 
Thus, a simple way to mark all the longest overlaps of strings with prefix $v$ in $O(||P||)$ time, is as follows: \\

\noindent
\fbox{\parbox{\textwidth}{
\textbf{Mark$\cal H$ for $ov(\cdot,v)$:} \\
{\em Traverse the ancestral path of $v$ from the root to $v$, storing for each leaf $x$ of $\cal A$ the last internal node $y$ having $x$ in ${\cal L}_y$. On reaching $v$, mark the stored internal nodes for each $x$.} 
}}

	\begin{algorithm}[tbh]
    	\caption{\textsc{Mark${\cal H}({\cal A},{\cal L})$}}
    	\label{alg:HOG}
    	\DontPrintSemicolon
    	\BlankLine
\lForEach{internal node $v$ of ${\cal A}$}{	$in{\cal H}[v]\gets 0$ 
    	\tcp*[f]{Flag for membership in ${\cal H}$}
    	}
\lForEach{leaf $v$ of ${\cal A}$}{
    	$in{\cal H}[v]\gets 1$
    	\tcp*[f]{Leaves implicitly in  ${\cal H}$}
    	}
    	$in{\cal H}[\epsilon]\gets 1$
    	\tcp*{Root implicitly in $\cal H$}
    	
    	\BlankLine
\lForEach{leaf $v$ of ${\cal A}$}{
    	${S}_v\gets \emptyset$ 
    	\tcp*[f]{Stack of exposed suffix}}

	\BlankLine
    	\ForEach(\tcp*[f]{Compute all $in{\cal H}[v]$})
    	{node $v\in {\cal A}$ in DFS order} {
        	\If{internal node $v$ first visited}{
        	    \lForEach(\tcp*[f]{Expose $v$ on stacks of ${\cal L}_v$}){leaf $x$ in ${\cal L}_v$}
        	    {Push $v$ on ${S}_x$
        	    }
        	    }
        	\If{internal node $v$ last visited}{
        	   \lForEach(\tcp*[f]{Remove $v$ from stacks of ${\cal L}_v$}){leaf $x$ in ${\cal L}_v$}
        	    {Pop $v$ from $S_x$
        	    }
        	}
        	\If{leaf $v$ visited}{
        	    \ForEach{leaf node $x$}
        	    {
        	    \If{$S_x \neq \emptyset$}{$in{\cal H}[$Top of $S_x]\gets 1$
        	    \tcp*[f]{Mark $ov(x,v)$}
        	    }
        	    }
        	}
    	}
    	\BlankLine
    	Return $in{\cal H}$\;
    	
	\end{algorithm}
	
\paragraph*{Marking all nodes in $Ov(P)$}
We now describe how to perform this procedure for all leaves (prefix) together (see \Cref{alg:HOG}) using stacks to keep track of the last encountered internal node for each leaf (suffix). The main reason behind using stacks is to avoid processing ${\cal L}_u$ multiple times (for different prefixes). For each internal node, we initialize the flag denoting membership in $\cal H$ to zero, whereas the root and leaves of $\cal A$ are implicitly in $\cal H$. For each leaf (suffix) we initialize an empty stack.
Now, we traverse $\cal A$ in DFS order (ignoring suffix links). As in the case for single leaf (prefix), the stack $S_x$ maintains the last internal node $v$ containing a leaf $x$ (suffix) in ${\cal L}_v$. This node  $v$ is added to the stack $S_x$ of the leaf $x$ (suffix) when $v$ is first visited by the traversal, and removed from the stack $S_x$ when it is last visited. This exposes the previously added internal nodes on the stack. Finally, on visiting a leaf $v$ (prefix), each non-empty stack $S_x$ of a leaf $x$ (suffix) exposes the internal node last added on its top, which is the longest overlap $ov(x,v)$ by \Cref{lem:condition}. We mark such internal nodes as being present in $\cal H$. The correctness follows from the same arguments used for the first approach. \\

In order to analyze the procedure we need to consider the processing of ${\cal L}_v$ and $S_x$ for all $v,x\in {\cal A}$, in addition to traversing ${\cal A}$. Since the total size of all ${\cal L}_v$ is $O(||P||)$, processing it twice (on the first and last visit of $v$) requires $O(||P||)$ time. This also includes the time to push and pop nodes from the stacks, requiring $O(1)$ time while processing ${\cal L}_v$. However, on visiting the leaf (prefix) by the traversal, we need to evaluate all $S_x$ and mark the top of non-empty stacks. Since we consider $n$ leaves (prefix), each processing all stacks of $n$ leaves (suffix), we require $O(n^2)$ time. For analyzing size, we need to consider only $S_x$ in addition to ${\cal L}_v$. Since the nodes in all $S_x$ are added once from some ${\cal L}_v$, the total size of all stacks $S_x$ is bounded by the size of all lists ${\cal L}_v$, i.e. $O(||P||)$ (as proved earlier). Thus, this procedure requires $O(||P||+n^2)$ time and $O(||P||)$ space to mark all nodes in $Ov(P)$.  

\newpage

	\begin{algorithm}[H]
    	\caption{\textsc{Mark${\cal H}({\cal A},{\cal L})$}}
    	\label{alg:optHOG}
    	\DontPrintSemicolon
    	\BlankLine
\lForEach{internal node $v$ of ${\cal A}$}{	$in{\cal H}[v]\gets 0$ 
    	\tcp*[f]{Flag for membership in ${\cal H}$}
    	}
\lForEach{leaf $v$ of ${\cal A}$}{
    	$in{\cal H}[v]\gets 1$
    	\tcp*[f]{Leaves implicitly in  ${\cal H}$}
    	}
$in{\cal H}[root]\gets 1$
\tcp*{Root implicitly in $\cal H$}
    	
    	\BlankLine
    	
\lForEach{leaf $v$ of ${\cal A}$}{
    	${S}_v\gets \emptyset$ 
    	\tcp*[f]{Stack of exposed suffix}}

{\color{blue}
${\cal S}\gets \emptyset$
        	\tcp*{List of stacks with unmarked tops}
\lForEach{leaf $v$ of ${\cal A}$}{
    	$in{\cal S}[v]\gets 0$ 
    	\tcp*[f]{Flag for membership of $S_v$ in $\cal S$}}
}

	\BlankLine
    	\ForEach(\tcp*[f]{Compute all $in{\cal H}[v]$})
    	{node $v\in {\cal A}$ in DFS order} {
        	\If{internal node $v$ first visited}{
        	    \ForEach(\tcp*[f]{Expose $v$ on stacks of ${\cal L}_v$}){leaf $x$ in ${\cal L}_v$}
        	    {Push $v$ on ${S}_x$\;
{\color{blue}
        	    \If(\tcp*[f]{Add $S_x$ to $\cal S$ if not present}){$in{\cal S}[x]=0$}{$in{\cal S}[x]\gets 1$\; Add $S_x$ to ${\cal S}$}
}        	    
        	    }
        	    }
        	\If{internal node $v$ last visited}{
        	    \ForEach(\tcp*[f]{Remove $v$ from stacks of ${\cal L}_v$}){leaf $x$ in ${\cal L}_v$}
        	    {Pop $v$ from $S_x$\;
{\color{blue}
        	    \uIf(\tcp*[f]{$S_x$ eligible in $\cal S$}){$S_x\neq \emptyset$ \textbf{and} $in{\cal H}[$Top of $S_x]=0$}
        	    {
        	     \If(\tcp*[f]{$S_x$ not present in $\cal S$}
        	    ){$in{\cal S}[x]=0$}{
        	    $in{\cal S}[x]\gets 1$\; Add $S_x$ to $\cal S$}}
        	
        	\Else(\tcp*[f]{$S_x$ either empty or with marked top}){
        	     \If(\tcp*[f]{$S_x$ present in $\cal S$}
        	    ){$in{\cal S}[x]=1$}{
        	    $in{\cal S}[x]\gets 0$\; Remove $S_x$ from $\cal S$}
        	}}}
}
        	\If{leaf $v$ visited}{
        	    \ForEach{ \color{red} $S_x\in {\cal S}$}
        	    {$in{\cal H}[$Top of $S_x]\gets 1$
        	    \tcp*{Mark $ov(x,v)$}
{\color{blue}        	    
     	    $in{\cal S}[x]\gets 0$\;
Remove $S_x$ from ${\cal S}$
\tcp*{Remove $S_x$ with marked top from $\cal S$}
}
        	    }
        	    
        	}
    	}
    	\BlankLine
    	Return $in{\cal H}$\;
    	
	\end{algorithm}

\paragraph*{Optimizing Mark${\cal H}$}
As described earlier, the only operation not bounded by $O(||P||)$ time is the marking of internal nodes, while processing the leaves (prefix) considering the stacks of all leaves (suffix). Note that this procedure is overkill as the same top of the stack can be marked again when processing different leaves (prefix), whereas total nodes entering and leaving stacks are proportional to total size of all ${\cal L}_u$, i.e., $O(||P||)$. Thus, we ensure that we do not have to process stacks of all leaves (suffix) on processing the leaves (prefix) of ${\cal A}$, and instead, we only process those stacks which were not processed earlier to mark the same top. Note that the same internal node may be marked again when exposed in different stacks, but we ensure that it is not marked again while processing the same stack.

Consider \Cref{alg:optHOG} (showing modified code in {\color{red} red} and additions in {\color{blue} blue}), we maintain a doubly linked-list $\cal S$ of non-empty stacks whose tops are not marked. Now, whenever a new node is added to a stack, it clearly has an unmarked top, so it is added to ${\cal S}$. And when a node is removed from a stack, the stack is added to  ${\cal S}$ if the new top is not previously marked and stack in not already in $\cal S$. Similarly, if the stack is empty or has a previously marked top, it is removed from $\cal S$ if it was present in $\cal S$.  Since ${\cal S}$ is a list, its members are additionally maintained using flags $in{\cal S}$ for each stack corresponding to leaves (suffix) of $\cal A$, so that the same stack is not added multiple times in $\cal S$. Also, each stack in ${\cal S}$ maintains a pointer to its location in $\cal S$, so that it can be efficiently removed if required. Now, on processing the leaves (prefix) of ${\cal A}$, we only process the stacks in $\cal S$, marking their tops and removing them from $\cal S$. Clearly, stacks are added to $\cal S$ only while processing ${\cal L}_v$, hence overall we can mark $O(|{\cal L}_v|)$ nodes for all $v$, requiring total $O(||P||)$ time. And the time taken in removing stacks from $\cal S$ is bounded by the total size of all $S_x$, which is also $O(||P||)$. Thus, we can perform Mark$\cal H$ using optimal $O(||P||)$ time and space, which results in our main result (using \Cref{thm:mainHOG}).

\optHOG*

\subparagraph*{Remark: } The classical result for APSP~\cite{GusfieldLS92} (equivalent to our second stage) was optimized~\cite{Gusfield1997} to get {\em output-sensitive} $O(||P||+n')$ time (where $n'$ is number of pairs with non-zero overlap) by maintaining a list of non-empty stacks (similar to our list $\cal S$ of stacks with non-marked heads). However, their approach does not suffice for computing $\cal H$ optimallty as in the worst case $n'=O(n^2)>> O(||P||)$ .

\section{Conclusions}
\label{sec:conc}
Genome assembly is one of the most prominent problems in Bioinformatics, and it traditionally relies on de Bruijn graphs or Overlap graphs, each having limitations of either loss of information or quadratic space requirements. Hierarchical Overlap Graphs provide a promising alternative that may result in better algorithms for genome assembly. The previous results on computing these graphs were not scalable (due to the quadratic time-bound) or required complicated data structures (segment trees). Moreover, computing HOG in optimal time and space was mentioned as an open problem in both the previous results~\cite{CazauxR20,ParkCPR20}. We present a simple algorithm that achieves the desired bounds, using only elementary data structures such as stacks and lists. At its core, we present an improved algorithm for a special case of All Pairs Suffix Prefix problem. We hope our algorithm directly, or after further simplification, results in a greater adaptability of HOGs in developing better genome assembly algorithms.

\bibliography{main}

\end{document}